\documentclass[12pt, oneside]{amsart}

\usepackage{filecontents}

\usepackage[margin=1.2in]{geometry}

\usepackage{setspace}

\usepackage{hyperref}

\usepackage{graphicx}

\usepackage[round]{natbib}

\usepackage{amsmath}
\usepackage{amsfonts}
\usepackage{amssymb}
\usepackage{amsthm}
\usepackage{cleveref}
\usepackage{enumerate}
\usepackage{xcolor}

\newtheorem*{MaxTh}{Maximum Theorem}
\newtheorem*{theorem*}{Theorem}
\newtheorem{theorem}{Theorem}
\newtheorem*{lemma*}{Lemma}
\newtheorem{lemma}{Lemma}
\newtheorem*{proposition*}{Proposition}
\newtheorem{proposition}{Proposition}
\newtheorem*{corollary*}{Corollary}
\newtheorem{corollary}{Corollary}

\newtheorem*{fact*}{Fact}

\theoremstyle{definition}

\newtheorem{definition}{Definition}
\newtheorem{example}{Example}

\crefname{lemma}{Lemma}{Lemmas}

\title{A Maximum Theorem for Incomplete Preferences}

\author{Leandro Gorno \hspace{30pt} Alessandro T. Rivello\vspace{2pt}\\
\hspace{3pt} \textit{FGV EPGE} \hspace{85pt} \textit{FGV EPGE} \hspace{10pt} \vspace{2pt}\\
}

\date{First draft: January 2020. Current version: November 2021.}

\begin{document}

\onehalfspacing

\begin{abstract}
We extend Berge's Maximum Theorem to allow for incomplete preferences. We first provide a simple version of the Maximum Theorem for convex feasible sets and a fixed preference. Then, we show that if, in addition to the traditional continuity assumptions, a new continuity property for the domains of comparability holds, the limits of maximal elements along a sequence of decision problems are maximal elements in the limit problem. While this new continuity property for the domains of comparability is not generally necessary for optimality to be preserved by limits, we provide conditions under which it is necessary and sufficient.\vspace{10pt}\\
\noindent \textit{Keywords:} incomplete preferences, maximum theorem, maximal elements, continuity.\vspace{10pt}\\
\noindent \textit{JEL classifications:} C61, C62.
\end{abstract}

\maketitle

\section{Introduction}

An important issue arising in the study of models involving optimization is whether optimal choices depend continuously on parameters affecting the objective function and the constraints. In fact, continuity of optimal choices plays a key role in the standard fixed-point approach to establish existence of equilibrium both in games and in competitive economies. The main tool to obtain (upper hemi)continuity is the Maximum Theorem by \citet*{Berge1963}, which can be stated as follows:

\begin{MaxTh}
Let $X$ and $\Theta$ be topological spaces; let $u: X \times \Theta \to \mathbb{R}$ be a continuous function; let $K : \Theta \rightrightarrows X$ be a continuous and compact-valued correspondence. Then, the correspondence $M:\Theta \rightrightarrows X$ defined by setting $M(\theta):=\arg\max_{x \in K(\theta)}u(x,\theta)$ for each $\theta \in \Theta$ is upper hemicontinuous and compact-valued.
\end{MaxTh}

This result can be easily modified to dispense with utility functions and deal directly with complete (continuous) preferences and also with incomplete preferences with open asymmetric parts (see \citet*{Walker1979}). However, to the best of our knowledge, none of the existing generalizations of the Maximum Theorem applies to the most prominent types of incomplete preferences such as Pareto orderings, preferences over lotteries admitting an expected multi-utility representation as in \citet*{Dubra2004}, or ordinal preferences possessing a multi-utility representation as in \citet*{Evren2011}. 

This is unfortunate for at least two reasons. First, completeness is not a mere technical assumption but one with significant behavioral implications. As a simple example, suppose that Ann is a consumer who has (monotone) preferences over two goods and chooses bundles $(1,0)$ and $(0,1)$ when prices are equal. Consider a 50\% increase in the price of good 2. If Ann's preferences are complete, her observed choices imply that Ann must be indifferent between $(1,0)$ and $(0,1)$ and, as a result, we would expect her to choose $(1,0)$ always at the new prices. Without completeness, however, Ann need not be indifferent between $(1,0)$ and $(0,1)$ and we can rationalize additional choices, such as $(0,0.5)$. It follows that, for certain applications, allowing for incomplete preferences might be necessary to accommodate the empirical evidence. 

Second, completeness is generally hard to defend on normative grounds, particularly when modeling decisions under uncertainty,\footnote{The main argument supporting the widespread adoption of the completeness axiom seems to be tractability. von Neumann and Morgestern themselves stated that: ``\textit{We have conceded that one may doubt whether a person can always decide which of two alternatives (...) he prefers. (...) If the general comparability assumption is not made, (...) a mathematical theory is still possible. It leads to what may be described as a many-dimensional vector concept of utility. This is a more complicated and less satisfactory set-up, but we do not propose to treat it systematically at this time.}'', \citet*{vNM1953}, p. 28-29.} and often fails for theoretical reasons in social choice contexts.

Although obtaining a maximum theorem for incomplete preference is desirable, the following example shows that it is not a trivial task:

\begin{example}
\label{ex:linear_consumer}
Ann chooses among bundles of two goods subject to a standard budget constraint. Her preferences are incomplete and can be represented with two utilities $u(q_1, q_2) = q_1 + q_2$ and $v(q_1, q_2) = q_1 + 2 q_2$ (\textit{i.e.}, a bundle $q=(q_1, q_2)$ is considered at least as good as another bundle $q'=(q'_1, q'_2)$ if and only if $u(q) \geq u(q')$ and $v(q) \geq v(q')$). The price of good 1 is normalized to $p_1=1$ and there is a decreasing sequence of prices for good 2, $p_{2,n} = 1+1/n$. As illustrated in Figure~\ref{fig:consumer}, if Ann's wealth is $w=1$, for each $n \in \mathbb{N}$, the bundle $A = (1,0)$ is optimal at prices $(p_1, p_{2,n})$: there is no feasible bundle that Ann strictly prefers to $A=(1,0)$. However, in the limit $n\to +\infty$, the bundle $A=(1,0)$ is no longer optimal because Ann strictly prefers the bundle $B=(0,1)$, which is feasible at prices $(p_1, p_{2,\infty}) = (1,1)$.
\end{example}

\begin{figure}[h]
\label{fig:consumer}
\includegraphics[scale=0.55]{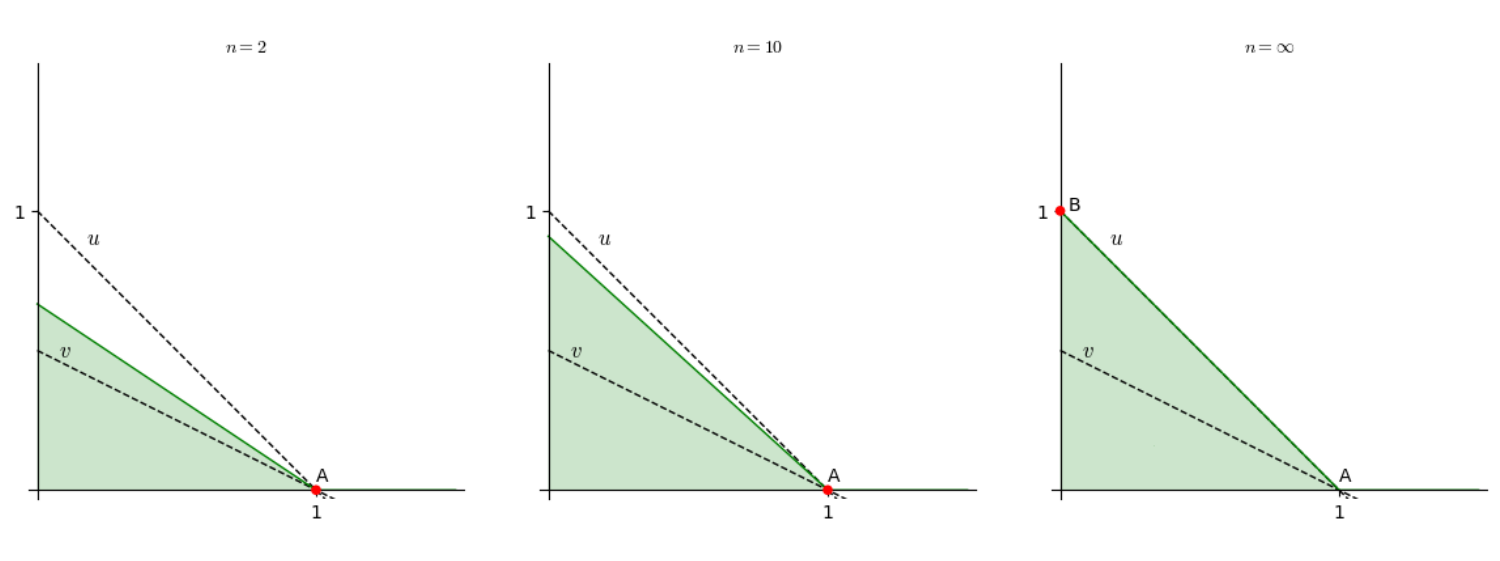}
\caption{Ann's problem for $n=2,10,\infty$.}
\end{figure}

The example above is fairly simple and suggests that, when weak preferences are continuous but incomplete, we should not be surprised to find sequences of maximal elements that converge to suboptimal alternatives. With this observation in mind, the main technical contributions of the present paper are: 1) to provide conditions which ensure that limits of optimal choices are optimal in the limit problem, and 2) to shed light on how the nature of preference incompleteness may interfere with the preservation of optimality when taking limits.

The rest of the paper is organized as follows. We present basic definitions in Section~\ref{section:prelims}. We state and prove a simple Maximum Theorem for incomplete preferences in Section~\ref{section:simple_MT}. This result is somewhat restrictive since it requires a fixed preference and convex feasible sets. In Section~\ref{section:MT}, we establish a more general Maximum Theorem based on a new continuity condition for the domains of comparability of the preferences involved. In Section~\ref{section:necessity}, we investigate assumptions under which this new condition is not only sufficient but also necessary. Finally, Section~\ref{section:literature} briefly relates our results to the existing literature and Section~\ref{section:discussion} offers some concluding remarks. We relegate all proofs to Appendix~\ref{section:proofs}.

\section{Preliminaries}
\label{section:prelims}
Let $(X, d)$ be a metric space. In this paper, a \textit{preference}, generically denoted by $\succsim$, is a reflexive and transitive binary relation on $X$. As usual, $\sim$ and $\succ$ denote the symmetric and asymmetric parts of $\succsim$, respectively. The  \textit{indifference classes} of $\succsim$ in $A \subseteq X$ are the sets of the form $\left\{y\in A\middle| y \sim x\right\}$ for some $x \in A$.

$\succsim$ is \textit{complete on} $A \subseteq X$ if either $x \succsim y$ or $y \succsim x$ holds for all $x,y\in A$ ($\succsim$ is \textit{complete} if it is complete on $X$). The set $A$ is a $\succsim$-\textit{domain} if $\succsim$ is complete on $A$. If $A \subseteq B \subseteq X$ and $A$ is a $\succsim$-domain such that no $\succsim$-domain contained in $B$ strictly contains $A$, then $A$ is a \textit{maximal $\succsim$-domain relative to $B$}. Denote by $\mathcal{D}\left(\succsim, B\right)$ the collection of all maximal $\succsim$-domains relative to $B$.

The alternative $x \in A$ is $\succsim$\textit{-maximal in} $A$ if, for every $y \in A$, $y \succsim x$ implies $x \succsim y$. The set of all alternatives that are $\succsim$-maximal in $A$ is denoted by $\mathcal{M}ax(\succsim, A)$. Analogously, $x \in A$ is $\succsim$\textit{-minimal in} $A$ if, for every $y \in A$, $x \succsim y$ implies $y \succsim x$. The set of all $\succsim$-minimal alternatives in $A$ is denoted by $\mathcal{M}in(\succsim, A)$.

A preference is \textit{continuous} if it is a closed subset of $X \times X$. Let $\mathcal{P}$ be the collection of continuous preferences on $X$. A set $\mathcal{U} \subseteq \mathbb{R}^X$ is a \textit{multi-utility representation} for preference $\succsim$ whenever, for every $x, y \in X$, $x \succsim y$ holds if and only if $u(x) \geq u(y)$ for all $u \in \mathcal{U}$.

Let $\mathcal{K}_X$ be the collection of nonempty compact subsets of $X$. Consider both $\mathcal{K}_X$ and $\mathcal{P}$ equipped with the Hausdorff metric topology derived from $X$ and $X \times X$, respectively. 
For any sequence $\left\{\mathcal{A}_n\right\}_{n \in \mathbb{N}}$ of nonempty subsets of $\mathcal{K}_X$, denote by $LS_{n \to +\infty} \mathcal{A}_n$ the collection of accumulation points of all sequences $\left\{A_n\right\}_{n \in \mathbb{N}}$, where $A_n \in \mathcal{A}_n$ for each $n \in \mathbb{N}$.

\section{A Simple Maximum Theorem}
\label{section:simple_MT}

In this section, we introduce a new continuity condition that is compatible with interesting classes of incomplete preferences and allows us to prove a simple Maximum Theorem. Throughout this section, we will assume that $X$ is convex.

\begin{definition}
A preference $\succsim$ is \textit{midpoint continuous} if, for every $x,y \in X$ satisfying $y \succ x$, there exists $\alpha \in [0,1)$ and open sets $V, W \subseteq X$ such that $\alpha x+(1-\alpha)y \in V$, $x \in W$, and $z' \succ x'$ for every $(z',x') \in V \times W$.
\end{definition}

Midpoint continuity is a substantive restriction on preferences. For instance, the standard vector order in $\mathbb{R}^2$ is not midpoint continuous. This does not mean that midpoint continuity is too restrictive for our purposes. After all, the standard vector order $\geq$ in $\mathbb{R}^2$ also violates the conclusion of the Maximum Theorem.\footnote{It is easy to construct a convergent sequence $\left(K_n, x_n\right)_{n \in \mathbb{N}}$, where, for each $n \in \mathbb{N}$, $K_n$ is a nonempty, convex, and compact subset of $\mathbb{R}^2$ and $x_n \in \mathcal{M}\left(\geq, K_n\right)$, with the property that $\lim_{n \to +\infty} x_n \not \in \mathcal{M}\left(\geq, \lim_{n\to+\infty}K_n\right)$.} 

In fact, for complete preferences, midpoint continuity is strictly weaker than continuity. To see that continuity implies midpoint continuity, note that the asymmetric part of a complete and continuous preference is open in $X \times X$. As a result, we can always take $\alpha = 0$ to satisfy the definition of midpoint continuity. To see that the converse implication does not hold, it suffices to consider the preference on $X=[0,1]$ represented by the utility function
\[
u(x) =
\begin{cases}
x & x \in [0,1/2)\\
2-x & x \in [1/2,1],
\end{cases}
\]
which is midpoint continuous but not continuous.

Without completeness, midpoint continuity and continuity are independent axioms. The incomplete preference on $X=[0,1]$ given by $\succsim \hspace{2pt} := [0,1)^2\cup \{(1,1)\}$ is midpoint continuous but not continuous. Moreover, both the natural vector order and the preference in Example~\ref{ex:linear_consumer} are continuous but not midpoint continuous.

The following result establishes that a significant class of incomplete preferences satisfies midpoint continuity:

\begin{proposition}
\label{prop:midpoint}
If $\succsim$ admits a finite multi-utility representation $\mathcal{U} \subseteq \mathbb{R}^X$ in which each $u \in \mathcal{U}$ is continuous and strictly quasiconcave, then $\succsim$ is midpoint continuous.
\end{proposition}

The requirement that $\succsim$ admits a \textit{finite} multi-utility representation in Proposition~\ref{prop:midpoint} cannot be dispensed with, for there are preferences admitting a countable multi-utility representation satisfying the remaining assumptions in Proposition~\ref{prop:midpoint} but violating midpoint continuity. 

Using the concept of midpoint continuity, we can establish the first major result of this paper, a simple Maximum Theorem:

\begin{theorem}
\label{th:simple_maximum}
Let $\succsim$ be a midpoint continuous preference and let $\left\{\left(K_n, x_n\right)\right\}_{n \in\mathbb{N}}$ be a sequence in $\mathcal{K}_X \times X$ such that
\begin{enumerate}
\item $\left\{\left(K_n, x_n\right)\right\}_{n \in \mathbb{N}}$ converges to $\left(K, x\right) \in \mathcal{K}_X \times X$ as $n\to +\infty$.
\item $x_n \in \mathcal{M}ax\left(\succsim, K_n\right)$ for every $n \in \mathbb{N}$.
\item $K_n$ is convex for every $n \in \mathbb{N}$.
\end{enumerate}
Then, $x \in \mathcal{M}ax\left(\succsim, K\right)$.
\end{theorem}

As we mentioned above, there are examples of incomplete preferences that are continuous and fail to satisfy a Maximum Theorem. It follows that midpoint continuity cannot be substituted for continuity in Theorem~\ref{th:simple_maximum}.

In the next example, we apply Proposition~\ref{prop:midpoint} and Theorem~\ref{th:simple_maximum} to establish the upper hemicontinuity of the set of Pareto efficient allocations with respect to the agent's endowments.

\begin{example}
Consider an environment with $L \in \mathbb{N}$ goods and $N \in \mathbb{N}$ agents. Here $X =\mathbb{R}_+^{L \times N}$ denotes the set of all conceivable allocations or social outcomes. Each agent has a continuous and strictly convex preference on $\mathbb{R}_+^L$.\footnote{A preference is said to be \textit{strictly convex} if, for every $x,y,z \in X$, $y \succsim x$, $z \succsim x$, $y \ne z$, and $\alpha \in (0,1)$ imply $\alpha y + (1-\alpha) z \succ x$. See \citet*{MWG1995}, p. 44.} The Pareto relation for this environment admits a finite multi-utility representation composed by continuous and strictly quasiconcave utilities. It thus follows from Proposition~\ref{prop:midpoint} that the Pareto relation on $X$ is midpoint continuous. Now consider a sequence of endowments $\left\{\omega_n\right\}_{n \in \mathbb{N}}$ that converges to $\omega$ and a sequence of allocations $\left\{x_n\right\}_{n \in \mathbb{N}}$ that converges to an allocation $x$. Denote by $\mathcal{E}_n$ and $\mathcal{E}$ the exchange economies associated by endowments $\omega_n$ and $\omega$, respectively. In this environment, Theorem~\ref{th:simple_maximum} directly implies that, if each allocation $x_n$ is Pareto efficient in the exchange economy $\mathcal{E}_n$, then allocation $x$ is Pareto efficient in the exchange economy $\mathcal{E}$.
\end{example}

It is important to emphasize that Theorem~\ref{th:simple_maximum} assumes a fixed preference $\succsim$. Relaxing this requirement (\textit{i.e.}, allowing for a sequence of preferences converging to $\succsim$) would significantly broaden the applicability of the result. Unfortunately, the following example shows that Theorem~\ref{th:simple_maximum} becomes false with this modification.

\begin{example}
Let $X=[0,1]$ and let $\succsim_n$ be represented by $\mathcal{U}_n := \left\{u, v_n\right\}$, where $u(x)=x$ and $v_n(x)=\left(x-\frac{n+1}{2n}\right)^2$. It is easy to verify that $\left\{\succsim_n\right\}_{n \in \mathbb{N}}$ converges to $\succsim$, the preference represented by $\mathcal{U} := \left\{u, v\right\}$, where $v(x) = \left(x-\frac{1}{2}\right)^2$. Moreover, by Proposition \ref{prop:midpoint}, every preference considered is midpoint continuous. However, even though $0 \in \mathcal{M}ax\left(\succsim_n, X\right)$ for each $n \in \mathbb{N}$, $0 \not\in \mathcal{M}ax\left(\succsim, X\right)$.
\end{example}

\section{A General Maximum Theorem}
\label{section:MT}

Theorem~\ref{th:simple_maximum} assumes a fixed preference and convex feasible sets. The second major result of this paper replaces these requirements with a continuity condition on maximal domains of comparability:

\begin{theorem}
\label{th:maximum}
Let $\left\{\left(\succsim_n, K_n, x_n\right)\right\}_{n \in\mathbb{N}}$ be a sequence in $\mathcal{P}\times \mathcal{K}_X \times X$ such that
\begin{enumerate}
\item $\left\{\left(\succsim_n, K_n, x_n\right)\right\}_{n \in \mathbb{N}}$ converges to $\left(\succsim, K, x\right) \in \mathcal{P}\times \mathcal{K}_X \times X$ as $n\to +\infty$.
\item $x_n \in \mathcal{M}ax\left(\succsim_n, K_n\right)$ for every $n \in \mathbb{N}$.
\item $LS_{n \to +\infty} \mathcal{D}\left(\succsim_n, K_n\right) \subseteq \mathcal{D}\left(\succsim, K\right)$.
\end{enumerate}
Then, $x \in \mathcal{M}ax\left(\succsim, K\right)$.
\end{theorem}

Theorem~\ref{th:maximum} obtains an upper-hemicontinuity property similar to that of the $\arg\max$ correspondence in Berge's Maximum Theorem but using condition (3) instead of assuming the existence of a utility representation (which is significantly stronger). Roughly, condition (3) requires all limits of maximal $\succsim_n$-domains to be maximal $\succsim$-domains, relative to the relevant feasible sets. In the particular case in which every preference in the sequence $\left\{\succsim_n\right\}_{n \in \mathbb{N}}$ is complete, the limit preference $\succsim$ must also be complete and, thus, condition (3) holds trivially. However, condition (3) is also compatible with incomplete preferences, as the following examples illustrate.

\begin{example}
\label{ex:equality}
Let $X=[0,1]$ and consider the ``diagonal'' relation $\succsim \hspace{2pt} := \left\{(x,x)\middle| 0 \leq x \leq 1\right\}$. Then, condition (3) holds for every convergent sequence $\left\{K_n\right\}_{n \in \mathbb{N}}$.
\end{example}

\begin{example}
\label{ex:moving}
Let $X=[-1,1]$ and consider the sequence of preferences $\left\{\succsim_n\right\}_{n \in \mathbb{N}}$ given by
\[
\succsim_n \hspace{2pt} =\left\{(x,y) \in X \times X\middle| x \geq y \geq 1/n \vee 1/n \geq y \geq x\right\},
\]
its limit
\[
\succsim \hspace{2pt} =\left\{(x,y) \in X \times X\middle| x \geq y \geq 0 \vee 0 \geq y \geq x\right\},
\]
and a sequence $\left\{K_n\right\}_{n \in \mathbb{N}}$ in $\mathcal{K}_X$ converging to some $K \in \mathcal{K}_X$. The corresponding maximal domains are $\mathcal{D}\left(\succsim_n, K_n\right)=\left\{K_n\cap [-1,1/n], K_n\cap [1/n,1]\right\} \setminus \{\emptyset\}$ and $\mathcal{D}\left(\succsim, K\right)=\left\{K\cap [-1,0], K\cap [0,1]\right\} \setminus \{\emptyset\}$. Assuming that $0$ is not in the boundary of $K$, the sequences $\left\{K_n\cap [-1,1/n]\right\}_{n \in \mathbb{N}}$ and $\left\{K_n\cap [1/n,1]\right\}_{n \in \mathbb{N}}$ converge to $K\cap [-1,0]$ and $K\cap [0,1]$,  respectively (while these limits are not entirely trivial and do not hold in general, we omit the details of the convergence argument for the sake of brevity). Thus, provided that $0$ is not in the boundary of $K$, condition (3) holds.
\end{example}

\begin{example}
\label{ex:partition}
Suppose there is a finite partition $\mathcal{D}^*$ of $X$ such that $\mathcal{D}\left(\succsim_n, X\right) = \mathcal{D}^*$ for all $n \in \mathbb{N}$. Note that this assumption nests the case of complete preferences as the particular case in which $\mathcal{D}^*=\{X\}$. Convergence of preferences implies that $\mathcal{D}\left(\succsim, X\right) = \mathcal{D}^*$ as well. Moreover, since all maximal domains relative to $X$ are disjoint, we also have $\mathcal{D}\left(\succsim, K\right) = \left\{D \cap K\middle| D\in \mathcal{D}^*\right\}$ and $\mathcal{D}\left(\succsim_n, K_n\right) = \left\{D \cap K_n\middle| D\in \mathcal{D}^*\right\}$ for every $n \in \mathbb{N}$. We conclude that $LS_{n \to +\infty} \mathcal{D}\left(\succsim_n, K_n\right) \subseteq \mathcal{D}\left(\succsim, K\right)$ and condition (3) holds.
\end{example}

\section{A Characterization}
\label{section:necessity}

Even though condition (3) in Theorem~\ref{th:maximum} constitutes a general sufficient condition for optimality to be preserved by limits, it is not necessary:

\begin{example}
\label{ex:not_necessary}
Let $X = [0, 1]$. Consider the following preference
\[
\succsim \hspace{2pt} = \left\{(x,y) \in [0,0.5)^2 \middle| x=y\right\} \cup [0.5,1]^2.
\]
Consider the sequence $\left\{K_n\right\}_{n \in \mathbb{N}}$, where $K_n := \left[0.5-0.5/n, 1\right]$ for each $n \in \mathbb{N}$. On the one hand, $D_n = \{0.5-0.5/n\}$ is a maximal $\succsim$-domain relative to $K_n$, while the sequence $\left\{D_n\right\}_{n \in \mathbb{N}}$ converges to $D = \{0.5\}$, which is not a maximal $\succsim$-domain relative to $K := \lim_{n \to +\infty} K_n = [0.5,1]$. On the other hand, $\mathcal{M}ax\left(\succsim_n, K_n\right) = \mathcal{M}in\left(\succsim_n, K_n\right) = K_n$ for all $n\in \mathbb{N}$ and $\mathcal{M}ax\left(\succsim, K\right) = \mathcal{M}in\left(\succsim, K\right)=K$, so all convergent sequences composed by $\succsim$-maximal and $\succsim$-minimal elements in each $K_n$ converge to $\succsim$-maximal and $\succsim$-minimal elements in $K$.
\end{example}

However, condition (3) is indeed necessary and sufficient for maximal and minimal elements to be preserved by limits in more specific settings. In this section, we obtain a characterization by restricting attention to limit preferences that are antisymmetric (\textit{i.e.}, partial orders) and sets that are ``order dense".

Formally, a set $A \subseteq X$ is $\succsim$-\textit{dense} if, for every $x,y \in A$, $x \succ y$ implies that there exists $z \in A$ such that $x \succ z \succ y$. $\succsim$-dense sets are quite common in applications. For instance, if $\succsim$ is a preference over lotteries that admits an expected multi-utility representation\footnote{\citet*{Dubra2004} show that a preference over lotteries has an expected multi-utility representation if and only if it is continuous and satisfies the independence axiom.}, then every convex set of lotteries is $\succsim$-dense. We can now state the main result of this section.

\begin{theorem}
\label{th:full_maximum}
Denote by $\mathcal{G} \subseteq \mathcal{P}$ the collection of continuous partial orders on $X$. Let $\left\{\left(\succsim_n, K_n\right)\right\}_{n \in\mathbb{N}}$ be a converging sequence in $\mathcal{P} \times \mathcal{K}_X$ with limit $\left(\succsim, K\right) \in \mathcal{G} \times \mathcal{K}_X$ and such that, for every $n \in \mathbb{N}$, $K_n$ is a $\succsim_n$-dense set and all indifference classes of $\succsim_n$ in $K_n$ are connected. Then, $K$ is $\succsim$-dense. Moreover, the following are equivalent:
\begin{enumerate}
\item $LS_{n \to +\infty} \mathcal{D}\left(\succsim_n, K_n\right) \subseteq \mathcal{D}\left(\succsim, K\right)$
\item $LS_{n \to +\infty} \mathcal{M}ax(\succsim_n, K_n) \subseteq \mathcal{M}ax(\succsim, K)$ and \\$LS_{n \to +\infty} \mathcal{M}in(\succsim_n, K_n) \subseteq {\mathcal{M}in(\succsim, K)}$.
\end{enumerate}
\end{theorem}

Theorem~\ref{th:full_maximum} provides assumptions under which condition (3) in Theorem~\ref{th:maximum} is necessary and sufficient for all limits of maximal or minimal elements to be maximal or minimal, respectively. Note that condition (2) in Theorem \ref{th:full_maximum} requires that \textit{every} convergent sequence of $\succsim$-maximal (resp. $\succsim$-minimal) elements converges to a $\succsim$-maximal (resp. $\succsim$-minimal) element. Thus, we may find some sequences of maximal elements that converge to a maximal element even when condition (1) is not satisfied.

\begin{example}
Let $\succsim$ be the natural vector order on $X = [0,1]^2$ . $\succsim$ is a continuous partial order and $X$ is $\succsim$-dense. Now, for each $n\in \mathbb{N}$, consider
\[
K_n := \left\{(x_1, x_2) \in X \middle| x_2 \leq n(1-x_1)\right\}.
\]
For each $n \in \mathbb{N}$, $K_n$ is nonempty and compact, and $(1-1/n, 1) \in \mathcal{M}ax(\succsim, K_n)$. Moreover, $\lim_{n \to + \infty} K_n = K := [0,1]^2$ and $(1,1) \in \mathcal{M}ax(\succsim, K)$. However, $[0,1] \times \{0\}$ is a   maximal $\succsim$-domain relative to $K_n$ but not a maximal $\succsim$-domain relative to $K$. It follows from Theorem~\ref{th:full_maximum} that there must be some sequence of maximal elements with a limit that is not maximal. Indeed, $(1, 0) \in \mathcal{M}ax(\succsim, K_n)$ for every $n \in \mathbb{N}$, but $(1,0)  \notin \mathcal{M}ax(\succsim, K)$.
\end{example}

An immediate application of Theorem~\ref{th:full_maximum} is to show that, in the case of a fixed preference, antisymmetry and midpoint continuity combined with mild additional assumptions imply condition (3) in Theorem~\ref{th:maximum}.

\begin{corollary}
\label{cor:midpoint}
If, in addition to the conditions of Theorem~\ref{th:simple_maximum}, $\succsim$ is a partial order and there exists $\underline{x} \in \bigcap_{n \in \mathbb{N}}K_n$ such that $x \succsim \underline{x}$ for every $x \in \bigcup_{n \in \mathbb{N}}K_n$, then condition (3) in Theorem~\ref{th:maximum} holds.
\end{corollary}

This means that Theorem~\ref{th:maximum} does generalize Theorem~\ref{th:simple_maximum} for partial orders in contexts such as consumer theory (in which consuming zero of every good is always feasible and there is no worse bundle).

A clear limitation of Theorem~\ref{th:full_maximum} is that the limit preference $\succsim$ is required to be a partial order. Example~\ref{ex:not_necessary} above shows that the characterization does not hold without this assumption even if $X=[0,1]$.

Another drawback of Theorem~\ref{th:full_maximum} is that two of its assumptions, namely $\succsim_n$-denseness of the $K_n$ and connectedness of the relative indifference classes, refer to the specific sequence $\left\{\left(\succsim_n, K_n\right)\right\}_{n \in\mathbb{N}}$ under consideration. However, if we restrict attention to convex feasible sets and preferences over lotteries that admit an expected multi-utility representation\footnote{Let $C$ be a separable metric space of consequences and let $X$ be the set of (Borel) probability measures (lotteries) over $C$ equipped with the topology of weak convergence of probability measures. Similarly to \citet*{Dubra2004}, we say that a set $\mathcal{U}$ of bounded continuous functions $C \to \mathbb{R}$ constitutes an \textit{expected multi-utility representation} for $\succsim$ whenever, for every two lotteries $x,y \in X$, $x \succsim y$ is equivalent to $\int_C u(c) dx(c) \geq \int_C u(c)dy(c)$ for all $u \in \mathcal{U}$.}, these assumptions are automatically satisfied for all sequences $\left\{\left(\succsim_n, K_n\right)\right\}_{n \in\mathbb{N}}$. This is the content of the following corollary:

\begin{corollary}
\label{cor:EMU}
Let $X$ be the space of (Borel) probability measures on a separable metric space equipped with the topology of weak convergence. Suppose further that:
\begin{enumerate}
\item For each $n \in \mathbb{N}$, $K_n \in \mathcal{K}_X$ is convex and $\succsim_n$ is a preference that admits an expected multi-utility representation,
\item $\succsim$ is a partial order that admits an expected multi-utility representation.
\end{enumerate}
Then, the equivalence in the conclusion of Theorem~\ref{th:full_maximum} holds.
\end{corollary}

\section{Related literature}
\label{section:literature}

Considerable work has been devoted to the study of incomplete preferences.\footnote{The list is long. A few examples in chonological order are \citet*{Aumann1962},\citet*{Peleg1970},\citet*{Ok2002}, \citet*{Dubra2004}, \citet*{Eliaz2006}, \citet*{Dubra2011}, \citet*{Evren2011}, \citet*{Ok2012}, \citet*{Evren2014}, \citet*{Riella2015}, \citet*{Gorno2017}, and \citet*{Gorno2018}.}. Despite the fact that continuous weak preferences are one of the two central classes of preferences in this literature, to the best of our knowledge, this is the first paper to provide positive results on the continuity properties of maximal elements for this class.

Continuity of optimal choices is an important problem and has been studied extensively. The central result is the Maximum Theorem in \citet*{Berge1963}, which is concerned with the behavior of value functions and the maximizers that attain them as parameters change continuously. Our results, in particular Theorem~\ref{th:maximum}, extend this work by allowing for incomplete preferences. Even though its key condition is trivially satisfied when preferences are complete, Theorem~\ref{th:maximum} is not truly a generalization of the original Maximum Theorem for three reasons. First, we assume that $X$ is a metric space, whereas Berge's result is proven in a general topological space. Second, since assuming that preferences admit a utility representation would imply completeness, our results focus exclusively on maximal elements and make no statements about value functions. Finally, we equip the space of preferences with the Hausdorff metric topology, which is finer than the topology implied by the continuity of the utility representation with respect to the parameter assumed by the original Maximum Theorem.\footnote{In Berge's framework, it would be natural to say that a sequence of preferences $\left\{\succsim_n\right\}_{n \in \mathbb{R}}$ converges to $\succsim$ if there exists a sequence $\left\{u_n\right\}_{n \in \mathbb{N}}$ such that, for each $n \in \mathbb{N}$, $u_n :X \to \mathbb{R}$ represents $\succsim_n$, and $\left\{u_n\right\}_{n \in \mathbb{N}}$ converges uniformily to a limit $u : X \to \mathbb{R}$ which represents $\succsim$.} For instance, every sequence of (complete) preferences converges to total indifference in that topology but not necessarily in ours.

\citet*{Walker1979} proves a generalized maximum theorem for a strict relation $\succ_\theta$ that depends on a parameter $\theta \in \Theta$ and has open graph (as a correspondence $\Theta \rightrightarrows X \times X$). Lemma 9 in \citet*{Evren2014} shows that if $X$ is a space of lotteries and $\succ$ is open, then $\mathcal{M}ax\left(\succsim,K\right)$ is relatively closed in $K$ and the correspondence $K \rightrightarrows \mathcal{M}ax\left(\succsim,K\right)$ is upper hemicontinuous, even if $K$ is neither convex nor compact. However, neither of these two results bears significance for the class of incomplete continuous weak preferences considered in the present paper. The reason is that, as long as $X$ is connected, every incomplete continuous weak preference that has an open strict part must be trivial (see \citet*{Schmeidler1971}) and, as a result, satisfy $\mathcal{M}ax\left(\succsim, K\right)=K$ for every $K \subseteq X$.

\section{Discussion}
\label{section:discussion}

The present paper provides three major results that expand the scope of Berge's Maximum Theorem to allow for incomplete preferences. Theorem \ref{th:simple_maximum} is based on a simple continuity condition, but its applicability is somewhat limited since it requires convex feasible sets and a fixed preference.

Theorem \ref{th:maximum} does not have the aforementioned limitations. The result depends crucially on its condition (3), a form of upper hemicontinuity of the mapping between preferences-feasible sets pairs and the corresponding collection of maximal domains of comparability. Since \citet*{Gorno2018} shows that every maximal element is the best element in some maximal domain and vice-versa, convergence of maximal domains permits the application of a Berge-type of argument to ensure the convergence of maximal elements through the convergence of local best elements, where the term ``local" here means ``relative to a maximal domain''. 

Finally, Theorem \ref{th:full_maximum} describes a more specific setting in which condition (3) in Theorem \ref{th:maximum} is necessary and sufficient for minimality and maximality to be preserved when taking limits. 

We believe that these results constitute a step forward towards understanding convergence of maximal elements without completeness and open at least three avenues for future research. First, the abstract nature Theorem~\ref{th:maximum} suggests to look for additional sets of assumptions which are sufficient for its condition (3) to hold. Second, the equivalence in Theorem~\ref{th:full_maximum} might be true under weaker assumptions. Third, we currently do not know whether our results remain true in a general topological space (which is the environment in which Berge's Maximum Theorem is formulated).

\appendix

\section{Proofs}
\label{section:proofs}

\begin{proof}[Proof of Proposition~\ref{prop:midpoint}]
Take $x,y \in X$ such that $x \succ y$. By definition of multi-utility representation, we have $u(x)\geq u(y)$ for all $u \in \mathcal{U}$. Define $z:= (1/2) x + (1/2) y$. Since each $u \in \mathcal{U}$ is strictly quasi-concave, $u(z)>u(y)$ holds for all $u\in \mathcal{U}$. For each $u \in \mathcal{U}$, there are open sets $V_u, W_u$ such that $(z,y) \in V_u \times W_u$ and $u(z') > u(y')$ for all $(z',y') \in V_u \times W_u$. Define $V := \bigcap_{u \in \mathcal{U}} V_u$ and $W := \bigcap_{u \in \mathcal{U}} W_u$. Clearly $(z,y) \in V \times W$, so $V$ and $W$ are nonempty. Moreover, since $\mathcal{U}$ is finite, $V$ and $W$ are open. By construction, we have $u(z')>u(y')$ for all $z' \in V$, $y' \in W$, and $u \in \mathcal{U}$. Since $\mathcal{U}$ is a multi-utility representation, $(z',y') \in V \times W$ implies $z' \succ y'$, showing that $\succsim$ satisfies midpoint continuity.
\end{proof}

\begin{proof}[Proof of Theorem~\ref{th:simple_maximum}]
Suppose, seeking a contradiction, that $x \notin \mathcal{M}ax\left(\succsim, K\right)$. Then, there exists $y \in K$ such that $y \succ x$. By midpoint continuity, there must exist $\alpha \in [0,1)$ and open sets $V, W \subseteq X$ satisfying $\left(\alpha x+(1-\alpha)y, x\right) \in V \times W$ and $z' \succ x'$ for all $(z',x') \in V \times W$. Define $z:=\alpha x+(1-\alpha) y$. Note that, being the limit of a sequence of convex sets, $K$ must be convex. It follows that $z \in K$. Since $x = \lim_{n\to+\infty} x_n$, there must be $N_1 \in \mathbb{N}$ such that $x_n \in W$ for all $n \geq N_1$. Since $K = \lim_{n \to +\infty} K_n$ and $z \in K$, there must be $N_2 \in \mathbb{N}$ such that $K_n \cap V \ne \emptyset$ for all $n \geq N_2$. Take $N := \max\{N_1, N_2\}$. Then, for all $n \geq N$, there exists $z_n \in K_n$ such that $z_n \succ x_n$. This implies $x_n \not \in \mathcal{M}ax\left(\succsim, K_n\right)$, a contradiction. We conclude that $x \in \mathcal{M}ax\left(\succsim, K\right)$.
\end{proof}

The proof of Theorem \ref{th:maximum} requires some results about Hausdorff convergence. In the following two lemmas $(M,d)$ is any metric space and $\mathcal{K}_M$ (resp. $\mathcal{F}_M$) is the collection of all nonempty compact (resp. closed) subsets of $M$.

\begin{lemma}
\label{lemma:subsequence}
Let $\{K_n\}_{n \in \mathbb{N}}$ be a convergent sequence in $\mathcal{K}_M$ with limit $K \in \mathcal{K}_M$. Then, every sequence $\{A_n\}_{n \in \mathbb{N}}$ in $\mathcal{K}_M$ such that $A_n \subseteq K_n$ for all $n\in \mathbb{N}$ has a subsequence which converges to a nonempty compact subset of $K$.
\end{lemma}

\begin{proof}
Let $d^H : \mathcal{K}_M \times \mathcal{K}_M \to \mathbb{R}_+$ denote the Hausdorff distance and let $\mathcal{K}_K$ be the collection of all nonempty compact subsets of $K$. Note that $\left(\mathcal{K}_M, d^H\right)$ is a metric space, $\mathcal{K}_K$ is compact in the (relative) Hausdorff metric topology, and $d^H(A_n, \cdot)$ is continuous on $\mathcal{K}_K$. For each $n \in \mathbb{N}$, let $B_n \in \arg\min_{\tilde{B} \in \mathcal{K}_K} d^H(A_n, \tilde{B})$. I now claim that, for each $n \in \mathbb{N}$, we have
\[
d^H(A_n, B_n) \leq d^H(K_n, K).
\]
To prove this claim, note that, since $\{y\} \in \mathcal{K}_K$ for all $y \in K$, we have
\[
d^H(A_n, B_n) \leq d^H(A_n, \{y\}) = \max_{x \in A_n}d(x,y)
\]
for all $y \in K$. Defining $y^*(x) \in \arg\min_{y \in K} d(x,y)$ for each $x \in A_n$, we have
\[
d^H(A_n, B_n) \leq \max_{x \in A_n} d(x,y^*(x)) = \max_{x \in A_n}\min_{y \in K} d(x,y)\phantom{............}
\]
\[
\leq \max_{x \in K_n}\min_{y \in K} d(x,y) \leq d^H(K_n, K)
\]
as desired. It follows that $\lim_{n \to +\infty} d^H(A_n , B_n) \leq \lim_{n \to +\infty} d^H(K_n , K_n) = 0$.

Since $\mathcal{K}_K$ is compact, the sequence $\{B_n\}_{n \in \mathbb{N}}$ has a convergent subsequence, say $\{B_{n_h}\}_{h \in \mathbb{N}}$. Let $A := \lim_{h \to +\infty} B_{n_h} \in \mathcal{K}_K$. Since $\lim_{n \to +\infty} d^H(A_n , B_n) = 0$ and $\lim_{h\to+\infty} B_{n_h} = A$, the triangle inequality $d^H(A_{n_h}, A) \leq d^H(A_{n_h}, B_{n_h})+d^H(B_{n_h}, A)$ implies $\lim_{h\to +\infty} d^H(A_{n_h}, A) = 0$. This means that $\{A_{n_h}\}_{h \in \mathbb{N}}$ converges to $A$, completing the proof.
\end{proof}

\begin{lemma}
\label{lemma:closed-convergence}
Denote by $\mathcal{F}_M$ the collection of all nonempty closed subsets of $M$ and let $\left\{\left(F_n, x_n\right)\right\}_{n \in \mathbb{N}}$ be a convergent sequence on $\mathcal{F}_M \times M$ with limit $\left(F, x\right) \in \mathcal{F}_M \times M$ and such that $x_n \in F_n$ for every $n \in \mathbb{N}$. 
Then, $x \in F$.
\end{lemma}

\begin{proof}
Suppose, seeking a contradiction, that $x \notin F$. Since $F$ is closed, there exists $\epsilon > 0$ such that $\left\{ y \in M \middle| d(x, y) \le \epsilon \right\} \cap F = \emptyset$. Hence, $\inf_{y \in F} d(x, y) > \epsilon/2$. By the triangule inequality, we have
\[
d(x, y) \le d(x, x_n) + d(x_n, y)
\]
for all $y \in F$ and all $n \in \mathbb{N}$. Therefore
\[
\inf_{y \in F} d(x, y) \le \inf_{y \in F} \left\{ d(x, x_n) + d(x_n, y) \right\} = d(x, x_n) + \inf_{y \in F} d(x_n, y) \leq d(x, x_n) + d^H(F_n, F)
\]
for all $n \in \mathbb{N}$, where we used
\[
\inf_{y \in F} d\left( x_n, y \right) \le \sup_{x \in F_n} \inf_{y \in F} d(x, y) \le d^H(F_n, F)
\]
Taking limits we conclude that $\inf_{y \in F} d(x, y) = 0$, a contradiction.
\end{proof}

\begin{lemma}
\label{lemma:Hausdorff_facts}
Consider a convergent sequence $\left\{\left(\succsim_n, K_n, x_n, y_n\right)\right\}_{n \in \mathbb{N}}$ in $\mathcal{P}\times \mathcal{K}_X \times X \times X$ with limit $\left(\succsim, K, x, y\right) \in \mathcal{P}\times \mathcal{K}_X \times X \times X$. Then:
\begin{enumerate}
\item $x_n \in K_n$ for all $n \in \mathbb{N}$ implies $x \in K$.
\item $x_n \succsim_n y_n$ for all $n \in \mathbb{N}$ implies $x \succsim y$.
\end{enumerate}
\end{lemma}

\begin{proof}
The first part follows from Lemma~\ref{lemma:subsequence} by taking $M=X$ and $A_n = \{x_n\}$ for each $n \in \mathbb{N}$, since every subsequence of $\{x_n\}_{n \in \mathbb{N}}$ converges to $x \in K$. The second part follows from Lemma~\ref{lemma:closed-convergence} by taking $M = X \times X$ and noting that $x_n \succsim_n y_n$ means $(x_n, y_n) \in \hspace{2pt} \succsim_n$.
\end{proof}

\begin{proof}[Proof of Theorem~\ref{th:maximum}]
For each $n \in \mathbb{N}$, there exists $D_n \in \mathcal{D}\left(\succsim_n, K_n\right)$ such that $x_n \in D_n$. Since $K_n \in \mathcal{K}_X$ converges to $K \in \mathcal{K}_X$ and $D_n$ is closed in $K_n$, we have $D_n \in \mathcal{K}_X$. By Lemma~\ref{lemma:subsequence}, there exists a convergent subsequence $(D_{n_h})_{h \in \mathbb{N}}$. As a result, there is no loss of generality in assuming that $\left\{D_n\right\}_{n \in \mathbb{N}}$ itself converges. Define $D := \lim_{n \to +\infty} D_n$. By Lemma ~\ref{lemma:Hausdorff_facts}, $x_n \in D_n$ for all $n \in \mathbb{N}$ and $\lim_{n \to +\infty} (D_n, x_n) = (D, x)$ together imply that $x \in D$. Moreover, condition (3) implies that $D \in \mathcal{D}\left(\succsim, K\right)$.
	
We now claim that $x$ is $\succsim$-maximal in $D$. To prove this, suppose, seeking a contradiction, that there exists $y \in D$ such that $y \succ x$. Since $\lim_{n \to +\infty} D_n = D$, there must exist a sequence $\left\{y_n\right\}_{n \in \mathbb{N}}$ such that $\lim_{n \to +\infty} y_n =y$ and $y_n \in D_n$ for every $n \in \mathbb{N}$. Moreover, since $\lim_{n \to +\infty}\succsim_n \hspace{2pt} = \hspace{2pt} \succsim$, by the second part of Lemma \ref{lemma:Hausdorff_facts}, there must exist $N \in \mathbb{N}$ such that $y_N \succ_N x_N$. This contradicts that $x_N$ is $\succsim_N$-maximal in $K_N$, as assumed.
	
Since $x$ is $\succsim$-maximal in the maximal $\succsim$-domain $D \in \mathcal{D}\left(\succsim, K\right)$, Theorem 1 of \citet*{Gorno2018} implies that $x$ is $\succsim$-maximal in $K$.
\end{proof}

The next three lemmas are used in the proof of Theorem~\ref{th:full_maximum}. In what follows, $K$ is an element of $\mathcal{K}_X$ and $\succsim$ is a continuous preference on $X$. A set $A \subseteq K$ \textit{has no exterior bound in} $K$ if, for every $x,y\in K$, $x \succsim A \succsim y$ implies $x,y \in A$.
	
\begin{lemma}
\label{lm:maxdomain_fullchar}
Assume that $K$ is $\succsim$-dense and the \textit{indifference classes} of $\succsim$ in $K$ are connected. Then, a subset of $K$ is a maximal $\succsim$-domain relative to $K$ if and only if it is a connected $\succsim$-domain relative to $K$ which has no exterior bound in $K$.
\end{lemma}
	
\begin{proof}
Since $K$ is compact and $\succsim \cap \left(K\times K\right)$ is a continuous preference on $K$ with connected indifference classes, the result follows from Theorem 4 in \citet*{GornoRivello2020}.
\end{proof}

\begin{lemma}
\label{lm:conv_connectset}
Let $\left\{K_n\right\}_{n \in \mathbb{N}}$ be a sequence in $\mathcal{K}_X$ such that $K_n$ is connected for every $n$ and $\lim_{n \to +\infty} K_n = K$, then $K$ is connected.
\end{lemma}
	
\begin{proof}
Suppose, seeking a contradiction, that $K$ is not connected. Then, there exist disjoint nonempty sets $A$ and $B$ which are closed in $K$ and satisfy $A \cup B = K$. For any $\epsilon > 0$, define $A^{\epsilon} := \left\{ x \in X \middle| d(x, A) < \epsilon \right\}$ and $\bar{A}^{\epsilon}$ as its closure. Define $B^{\epsilon}$ and $\bar{B}^{\epsilon}$ analogously. Define $K^{\epsilon} := A^{\epsilon} \cup B^{\epsilon}$ and $\bar{K}^{\epsilon}$ as its clousure.
		
Since $K$ is closed, $A$ and $B$ are also closed in $X$, which is a normal space. Then, there exists $\bar{\epsilon} > 0$ such that, for every $\epsilon \in (0, \bar{\epsilon}]$, we have $A^{\epsilon} \cap B^{\epsilon} = \emptyset$. Fix $\epsilon = \bar{\epsilon}/2$, then $\bar{A}^{\epsilon} \cap \bar{B}^{\epsilon} = \emptyset$. Because $\lim_{n \to +\infty} K_n = K$ there is $N_{\epsilon} \in \mathbb{N}$ such that $n \ge N_{\epsilon}$ implies $K_n \subseteq \bar{K}^{\epsilon}$. Now define $A_n := K_n \cap \bar{A}^{\epsilon}$ and $B_n := K_n \cap \bar{B}^{\epsilon}$. It is easy to see that $A_n \cap B_n = \emptyset$, $A_n \cup B_n = K_n$, and $A_n, B_n \in \mathcal{K}_X$. It follows that $K_n$ is not connected, a contradiction.
\end{proof}

\begin{lemma}
\label{lm:pref_connectset}
Let $\left\{\left(\succsim_n, K_n\right)\right\}_{n \in \mathbb{N}}$ be a converging sequence in $\mathcal{P} \times \mathcal{K}_X$ with limit $(\succsim, K) \in \mathcal{G} \times \mathcal{K}_X$. If, for every $n \in \mathbb{N}$, $K_n$ is $\succsim_n$-dense and all indifference classes of $\succsim_n$ in $K_n$ are connected, then $K$ is $\succsim$-dense.
\end{lemma}
	
\begin{proof}
Suppose, seeking a contradiction, that $K$ is not $\succsim$-dense. Then, there exist $x, y \in K$ such that $x \succ y$ and there is no $z \in K$ that satisfies $x \succ z \succ y$. Take $\left\{x_n\right\}_{n \in \mathbb{N}}$ and $\left\{y_n\right\}_{n \in \mathbb{N}}$ such that $x_n, y_n \in K_n$ for every $n \in \mathbb{N}$, $\lim_{n \to +\infty} x_n = x$, and $\lim_{n \to +\infty} y_n = y$. Define $M_n := \left\{ z \in K_n \middle| x_n \succsim_n z \succsim_n y_n \right\}$. Note that $M_n \in \mathcal{K}_X$ and $M_n \subseteq K_n$ for every $n \in \mathbb{N}$, so Lemma~\ref{lemma:subsequence} implies that $\left\{M_n\right\}_{n \in \mathbb{N}}$ has a convergent subsequence. Thus, we can assume without loss of generality that $\left\{M_n\right\}_{n \in \mathbb{N}}$ itself converges and define $M := \lim_{n \to +\infty} M_n$. 
		
We claim that $M_n$ is connected for each $n \in \mathbb{N}$. Suppose, seeking a contradiction, that $M_n$ is not connected for some $n \in \mathbb{N}$. Then there should exist disjoint nonempty sets $A$ and $B$ which are closed in $M_n$ and satisfy $A \cup B = M_n$. Without loss, assume that $x_n  \in A$. Since $B$ is compact and $\succsim_n$ is continuous, there exists at least one $\succsim_n$-maximal element in $B$, call it $\overline{x}_B$. Define $C := \left\{ z \in A \middle| z \succsim_n \overline{x}_B \right\}$. Note that $C$ is also compact and nonempty ($x_n \in C$), so we can take $\underline{x}_C$, one of its $\succsim_n$-minimal elements. We will now show that $\underline{x}_C \succ_n \overline{x}_B$. Define $I := \left\{ z \in K_n \middle| z \sim_n \bar{x}_B \right\}$. Since $I \subseteq M_n$, both $I \cap B$ and $I \cap A$ are closed sets which satisfy $\left(I \cap A\right) \cup \left(I \cap B\right) = I$ and $\left(I \cap A\right) \cap \left(I \cap B\right) = \emptyset$. Since $I$ is assumed to be connected and $\bar{x}_B \in I \cap B$ it must be that $I \cap A = \emptyset$, proving that $\underline{x}_C \succ_n \overline{x}_B$. Define $D := \left\{ z \in M_n \middle| \underline{x}_C \succ_n z \succ_n \overline{x}_B \right\}$. If there is $z \in D$, then $z \succ_n \overline{x}_B$ implies $z \notin B$. Moreover $\underline{x}_C \succ_n z \succ_n \overline{x}_B$ implies that $z \notin C$, so $z \notin A$ either. It follows that $D$ must be empty, which is a contradiction with $K_n$ being $\succsim_n$-dense. We conclude that $M_n$ is connected.
		
On the one hand, since $\left\{M_n\right\}_{n \in \mathbb{N}}$ is a sequence of connected sets in $\mathcal{K}_X$, $M$ is connected by Lemma \ref{lm:conv_connectset}. On the other hand, we claim that $M = \left\{x, y\right\}$. It is easy to see that $\left\{x, y\right\} \subseteq M$. To prove the other inclusion take any sequence $\left\{z_n\right\}_{n \in \mathbb{N}}$ converging to $z \in M$ and such that $z_n \in M_n$ for every $n \in \mathbb{N}$. Then $x_n \succsim_n z_n \succsim_n y_n$ implies $x \succsim z \succsim y$. But, since we initially assumed that there is no $z \in K$ that satisfies $x \succ z \succ y$, we must necessarily have that either $x \sim z$ or $z \sim y$. Furthermore, $\succsim$ antisymmetric implies that $x=z$ or $z=y$. Hence, $M$ must be equal to $\left\{x, y\right\}$.
		
We conclude that $M$ must simultaneously be connected and equal to $\left\{x, y\right\}$, which yields the desired contradiction.
\end{proof}

\begin{proof}[Proof of Theorem~\ref{th:full_maximum}]
$(1) \Rightarrow (2)$. The convergence of the maximal elements is a direct implication of Theorem \ref{th:maximum}. To see the convergence of the minimal elements let $\left\{y_n\right\}_{n \in \mathbb{N}}$ be a convergent sequence such that $\lim_{n \to + \infty} y_n = y$ and, for all $n \in \mathbb{N}$, $y_n$ is a $\succsim_n$-minimal element. If we define $\succsim^*_n \hspace{2pt} := \left\{(x,y) \in X \times X \middle| y \succsim_n x\right\}$, then $\mathcal{D}\left(\succsim^*_n, K_n\right) = \mathcal{D}\left(\succsim_n, K_n\right)$ for every $n \in \mathbb{N}$. Moreover, every $y_n$ is a $\succsim^*_n$-maximal. Thus we can apply Theorem \ref{th:maximum} to conclude that $y$ is a $\succsim^*$-maximal which is equivalent to it be a $\succsim$-minimal.
	
$(1) \Leftarrow (2)$. Take a sequence $\left\{D_n\right\}_{n \in \mathbb{N}}$ such that $D_n \subseteq K_n$, $D_n \in \mathcal{D}\left(\succsim_n, K_n\right)$, and $\lim_{n \to +\infty} D_n = D$. Since $\lim_{n \to +\infty} \succsim_n \hspace{2pt} = \hspace{2pt} \succsim$, $D$ is a $\succsim$-domain. For each $n \in \mathbb{N}$, $\succsim_n$ is a continuous preference such that $\succsim_n \cap \left(K_n \times K_n\right)$ has connected indifference classes and $K_n$ is $\succsim_n$-dense. Thus, by Lemma \ref{lm:maxdomain_fullchar}, $D_n$ is connected for every $n \in \mathbb{N}$. It follows that $D$ is also connected by Lemma \ref{lm:conv_connectset} and $\succsim$-dense by Lemma~\ref{lm:pref_connectset}.

We claim that $D$ has no exterior bounds. Suppose, seeking a contradiction, there is $x \in K$ such that $x \succsim\tilde{y}$ for every $\tilde{y} \in D$ and $x \notin D$. In fact, we must have $x \succ \tilde{y}$ because $\succsim$ is a partial order. For each $D_n$ take $y_n \in D_n$ such that $y_n \succsim_n z$ for all $z \in D_n$. By Theorem 1 in \citet*{Gorno2018} each $y_n$ is a maximal element in $K_n$. Note that $\left\{y_n\right\} \subset K_n$ and $\left\{y_n\right\}$ is compact for every $n \in \mathbb{N}$, thus by Lemma \ref{lemma:subsequence} there is no loss of generality in assuming that $\{y_n\}_{n\in \mathbb{N}}$ converges to some $y \in K$. Since $y_n \in D_n$ for every $n$ and $\left\{(D_n, y_n)\right\}_{n \in \mathbb{N}}$ converges to $(D,y)$, we must have $y \in D$. By hypothesis we have $x \succ y$ and $y$ is a maximal element in $K$, a contradiction. An analogous argument guarantees that there is no $x \in K$ such that $y \succsim x$ for every $y \in D$ and $x \notin D$. This means that $D$ has no exterior bounds.
	
Furthermore, since $\succsim$ is a partial order, $D$ contains all its indifferent alternatives. Thus, by Lemma \ref{lm:maxdomain_fullchar}, we conclude that $D \in \mathcal{D}\left(\succsim, K\right)$.
\end{proof}

\section*{Acknowledgments}
This paper is based on material contained in Chapter 2 of Alessandro Rivello's PhD dissertation at FGV EPGE. We thank Jos\'{e} Heleno Faro, Kazuhiro Hara, Paulo Klinger Monteiro, Lucas Maestri, Gil Riella, and participants at various seminars for their helpful comments. This study was financed in part by the Coordenação de Aperfeiçoamento de Pessoal de N\'{i}vel Superior - Brasil (CAPES) - Finance Code 001.

\bibliographystyle{abbrvnat}
\bibliography{refs}

\end{document}